\documentclass[11pt]{amsart}
\def\isdraft{1}

\usepackage{txfonts}
\usepackage{amsaddr} 
\usepackage{mathtools}
\usepackage{amsmath,amssymb,amsfonts,dsfont}
\usepackage{prettyref} 
\usepackage[utf8]{inputenc}
\usepackage{enumitem}
\usepackage[hyphens]{url}
\usepackage{tikz-cd}
\usepackage{tikz}
\usetikzlibrary{positioning}
\usetikzlibrary{arrows}
\usepackage{bussproofs}
\usepackage[mathscr]{euscript}
\usepackage{hyperref} 
\usepackage{amsthm}
\usepackage{theapa}
\usepackage{a4wide}
\usepackage[color=black,textcolor=white\if\isdraft0,disable\fi]{todonotes}
\usepackage{stackengine}
\usepackage{stmaryrd}
\usepackage{graphicx}
\usepackage{wrapfig}
\usepackage{float}
\graphicspath{{fig/}}

\newtheorem{theorem}{Theorem}

\newtheorem{proposition}[theorem]{Proposition}

\newtheorem{fact}[theorem]{Fact}

\theoremstyle{definition} 

\newtheorem{example}[theorem]{Example}

\newrefformat{cha}{Chapter \ref{#1}}
\newrefformat{sec}{Section \ref{#1}}
\newrefformat{tab}{Table \ref{#1}}
\newrefformat{fig}{Figure \ref{#1}}
\newrefformat{equ}{(\ref{#1})}
\newrefformat{app}{Appendix \ref{#1}}
\newrefformat{thm}{Theorem \ref{#1}}
\newrefformat{cor}{Corollary \ref{#1}}
\newrefformat{prop}{Proposition \ref{#1}}
\newrefformat{lem}{Lemma \ref{#1}}
\newrefformat{fact}{Fact \ref{#1}}
\newrefformat{obs}{Observation \ref{#1}}
\newrefformat{note}{Note \ref{#1}}
\newrefformat{idea}{Idea \ref{#1}}
\newrefformat{trivia}{Trivia \ref{#1}}
\newrefformat{def}{Definition \ref{#1}}
\newrefformat{not}{Notation \ref{#1}}
\newrefformat{con}{Convention \ref{#1}}
\newrefformat{rem}{Remark \ref{#1}}
\newrefformat{exa}{Example \ref{#1}}
\newrefformat{problem}{Problem \ref{#1}}
\newrefformat{claim}{Claim \ref{#1}}
\newrefformat{conjecture}{Conjecture \ref{#1}}
\newrefformat{exe}{Exercise \ref{#1}}
\newrefformat{alg}{Algorithm \ref{#1}}
\newrefformat{err}{Error \ref{#1}}
\newrefformat{que}{Question \ref{#1}}
\newrefformat{ite}{Item \ref{#1}}
\newrefformat{Q}{Q. \ref{#1}}
\newrefformat{warning}{Warning \ref{#1}}
\newrefformat{pseudocode}{Pseudocode \ref{#1}}

\title{Sequential decomposition of propositional logic programs\\(under construction!)}
\author{
    Christian Anti\'c
}
\address{
    christian.antic@icloud.com\\
    Vienna University of Technology\\
    Vienna, Austria
}

\begin{document}
\begin{abstract} 
	This paper studies the sequential decomposition of propositional logic programs by analyzing Green's relations $\mathcal{L,R,J}$---well-known in semigroup theory---in the logic programming domain.
\end{abstract}
\maketitle

\section{Introduction}

Rule-based reasoning is an essential part of human intelligence prominently formalized in artificial intelligence research via logic programs, which are formal constructs written in a rule-like sublanguage of predicate logic \cite<cf.>{Apt90,Apt97,Hodges94,Lloyd87,Makowsky87,Sterling94}. Logic programming is a well-established subfield of theoretical computer science and AI with applications to such diverse fields as expert systems, database theory, diagnosis, planning, learning, natural language processing, and many others \cite<cf.>{Baral03,Coelho88,Kowalski79,Pereira02}. The propositional Horn fragment consisting of propositional facts and rules is particularly relevant in database theory via the query language datalog \cite<cf.>{Ceri90} and in answer set programming \cite{Brewka11}, a prominent and successful dialect of logic programming incorporating negation as failure \cite{Clark78} and many other constructs such as aggregates \cite<cf.>{Faber04,Pelov04}, external sources \cite{Eiter05}, and description logic atoms \cite{Eiter08a}.

Describing complex objects as the composition of elementary ones is a common strategy in computer science and science in general as being able to assemble complex programs from simple ones is desirable from both the theoretical and the practical viewpoint. \citeA{Antic_i3} has recently introduced the sequential composition of propositional logic programs. Studying compositions and decompositions of programs gives insights into the modularity of programs \cite{Brogi99,Bugliesi94} both on a mathematical and on a programming level. For example, in the context of logic-based analogical reasoning and learning, programs having similar syntactic decompositions turn out to have similar semantics which is exploited for constructing novel programs from given ones \cite{Antic_i2}. That is, the similarity of programs is detected by comparing their respective decompositions. 

\citeA{OKeefe85} is probably the first to study the composition of logic programs and \citeA{Bugliesi94} study {\em modular logic programming} by defining algebraic operations on programs similar to \citeA{OKeefe85}. The relationship to sequential composition is discussed in more detail in \citeA[§7]{Antic_i3}.

The purpose of this paper is to initiate the study of the sequential {\em decomposition} of propositional logic programs by investigating \citeS{Green51} relations $\mathcal{L,R,J}$---which play a key role in semigroup theory \cite<cf.>{Howie03}---in the logic programming setting. To the best of our knowledge, this has not been considered before. Green's relations yield information about the {\em logcal} structure of the space of all propositional programs: given programs $P$ and $R$, Green's $\mathcal J$-relation is defined as $P\leq_\mathcal J R$ iff $P=(QR)S$, for some programs $Q,S$, in which case we can imagine $R$ to be a ``subprogram'' of $P$ in a way not expressible via union alone, that is, sequential composition allows us to split programs in a novel ``horizontal'' way. 

In a broader sense, this paper is a further step towards an algebraic theory of logic programming and in the future we plan to adapt and generalize the results of this paper to wider classes of programs, most importantly to first-, and higher-order logic programs \cite{Apt90,Chen93,Lloyd87,Miller12}, and non-monotonic logic programs under the stable model \cite{Gelfond91} or answer set semantics and extensions thereof \cite<cf.>{Baral03,Brewka11,Eiter09,Lifschitz19}---see \citeA{Antic_i6}.\\

\section{Propositional logic programs}

We recall the syntax and semantics of propositional logic programs \cite<cf.>{Lloyd87}.

\subsection{Syntax}

In the rest of the paper, $A$ denotes a finite alphabet of propositional atoms.

A ({\em propositional Horn logic}) {\em program} over $A$ is a finite set of {\em rules} of the form
\begin{align}\label{equ:r} 
    a_0\leftarrow a_1,\ldots,a_k,\quad k\geq 0,
\end{align} where $a_0,\ldots,a_k\in A$ are propositional atoms. It will be convenient to define, for a rule $r$ of the form \prettyref{equ:r}, $h(r):=\{a_0\}$ and $b(r):=\{a_1,\ldots,a_k\}$ extended to programs by $h(P):=\bigcup_{r\in P}h(r)$ and $b(P):=\bigcup_{r\in P}b(r)$. In this case, the {\em size} of $r$ is $k$ and denoted by $sz(r)$. A {\em fact} is a rule with empty body and a {\em proper rule} is a rule which is not a fact. We denote the facts and proper rules in $P$ by $f(P)$ and $p(P)$, respectively. We call a rule $r$ of the form \prettyref{equ:r} {\em Krom}\footnote{Krom rules were first introduced and studied by \cite{Krom67}.} if it has at most one body atom. 


Define the {\em dual} of $P$ by
\begin{align*} 
    P^d:=f(P)\cup\{b\leftarrow h(r)\mid r\in p(P),b\in b(r)\}.
\end{align*} Roughly, we obtain the dual of a program by reversing all the arrows of its proper rules.

\subsection{Semantics}

An {\em interpretation} is any set of atoms from $A$. We define the {\em entailment relation}, for every interpretation $I$, inductively as follows: (i) for an atom $a$, $I\models a$ if $a\in I$; (ii) for a set of atoms $B$, $I\models B$ if $B\subseteq I$; (iii) for a rule $r$ of the form \prettyref{equ:r}, $I\models r$ if $I\models b(r)$ implies $I\models h(r)$; and, finally, (iv) for a propositional logic program $P$, $I\models P$ if $I\models r$ holds for each rule $r\in P$. In case $I\models P$, we call $I$ a {\em model} of $P$. The set of all models of $P$ has a least element with respect to set inclusion called the {\em least model} of $P$ and denoted by $LM(P)$. We say that $P$ and $R$ are {\em equivalent} if $LM(P)=LM(R)$.

Define the {\em van Emden-Kowalski operator} \cite{vanEmden76} of $P$, for every interpretation $I$, by
\begin{align*} 
    T_P(I):=\{h(r)\mid r\in P:I\models b(r)\}.
\end{align*} This operator yields the immediate consequence derivable from a given interpretations and it is at the core of logic programming as its least and prefixed fixed points capture the semantics of a program. In fact, we have the following well-known operational characterization of models in terms of the van Emden-Kowalski operator.

\begin{proposition}\label{prop:prefixed} An interpretation $I$ is a model of $P$ iff $I$ is a prefixed point of $T_P$.
\end{proposition}

We call an interpretation $I$ a {\em supported model} of $P$ if $I$ is a fixed point of $T_P$. We say that $P$ and $R$ are {\em subsumption equivalent} \cite{Maher88} if $T_P=T_R$, denoted by $P\equiv_{ss}R$.

The {\em least fixed point} computation of $T_P$ is defined by
\begin{align*} 
    &T_P^0:=\emptyset,\\
    &T_P^{n+1}:=T_P(T_P^n),\\
    &T_P^\infty:=\bigcup_{n\geq 0}T_P^n.
\end{align*}

Finally, we have the constructive operational characterization of least models in terms the van Emden-Kowalski operator \cite{vanEmden76} given by
\begin{align*} 
    LM(P)= T_P^\infty.
\end{align*}

\section{Composition}\label{sec:Composition}

In this section, we recall the sequential composition of propositional logic programs as presented in \citeA{Antic_i3}.

In the rest of the paper, $P$ and $R$ denote propositional logic programs, and $I$ and $J$ denote interpretations over some joint alphabet $A$.

The rule-like structure of logic programs induces naturally a compositional structure as follows. We define the ({\em sequential}) {\em composition} of $P$ and $R$ by
\begin{align*} 
    P\circ R=\left\{h(r)\leftarrow b(S)\;\middle|\; r\in P, S\subseteq_{sz(r)} R,h(S)=b(r)\right\}.
\end{align*} We will write $PR$ in case the composition operation is understood. Roughly, we obtain the composition of $P$ and $R$ by resolving all body atoms in $P$ with the ``matching'' rule heads of $R$.

Notice that we can reformulate sequential composition as
\begin{align}\label{equ:bigcup} 
    P\circ R=\bigcup_{r\in P}(\{r\}\circ R),
\end{align} which directly implies right-distributivity of composition, that is,
\begin{align}\label{equ:(P_cup_Q)_circ_R} 
    (P\cup Q)\circ R=(P\circ R)\cup (Q\circ R)\quad\text{holds for all programs }P,Q,R.
\end{align} However, the following counter-example shows that left-distributivity fails in general:  
\begin{align*} 
    \{a\leftarrow b,c\}\circ(\{b\}\cup\{c\})=\{a\} \quad\text{and}\quad (\{a\leftarrow b,c\}\circ\{b\})\cup(\{a\leftarrow b,c\}\circ\{c\})=\emptyset.
\end{align*}

Define the {\em unit program} by
\begin{align*} 
    1_A:=\{a\leftarrow a\mid a\in A\}.
\end{align*} In the sequel, we will omit the reference to $A$. 

The space of all programs over some fixed alphabet is closed under sequential composition with the neutral element given by the unit program and the empty program serves as a left zero, that is, we have
\begin{align*} 
    P\circ 1=1\circ P=P \quad\text{and}\quad \emptyset\circ P=\emptyset.
\end{align*} Composition is, in general, {\em not} associative \cite<cf.>[Example 10]{Antic_i2}.

We can simulate the van Emden-Kowalski operator on a {\em syntactic} level without any explicit reference to operators via sequential composition in the sense that for any interpretation $I$,
\begin{align}\label{equ:T_P} 
    T_P(I)=PI.
\end{align}

As facts are preserved by composition and since we cannot add body atoms to facts via composition on the right, we have
\begin{align}\label{equ:IP=I} 
    IP=I.
\end{align}

Our next observation is that we can compute the heads and bodies of $P$ via
\begin{align}\label{equ:h(P)} 
    h(P)=PA \quad\text{and}\quad b(P)=p(P)^d A.
\end{align} Moreover, we have
\begin{align}\label{equ:h(PR)} 
    h(PR)\subseteq h(P) \quad\text{and}\quad b(PR)\subseteq b(R).
\end{align}

Given an interpretation $I$, we define
\begin{align*} 
    I^\ominus:=1^{A-I}\cup I \quad\text{and}\quad I^\oplus:=\{a\leftarrow(\{a\}\cup I)\mid a\in A\}.
\end{align*} It is not difficult to show that $PI^\ominus$ is the program $P$ where all occurrences of the  atoms in $I$ are removed from the rule bodies in $P$, that is, we have
\begin{align}\label{equ:PI^ominus} 
    PI^\ominus&=\{h(r)\leftarrow (b(r)-I)\mid r\in P\}.
\end{align} Similarly, $PI^\oplus$ is the program $P$ with the atoms in $I$ added to the rule bodies of all {\em proper} rules in $P$, that is, we have
\begin{align}\label{equ:PI^oplus}
    PI^\oplus&=f(P)\cup\{h(r)\leftarrow (b(r)\cup I)\mid r\in p(P)\}.
\end{align}

The left and right reducts can be represented via composition and the unit program by
\begin{align}\label{equ:^IP} 
    ^IP=1^I\circ P \quad\text{and}\quad P^I=P\circ 1^I.
\end{align} This implies
\begin{align}\label{equ:I_cap_J} 
    I\cap J={^J}I\stackrel{\prettyref{equ:^IP}}=(1^J)I.
\end{align}

We can compute the facts in $P$ via
\begin{align}\label{equ:f(P)} 
    f(P)=P\circ\emptyset.
\end{align} Moreover, we have
\begin{align}\label{equ:f(P)^ast} 
    P=f(P)\cup p(P)\stackrel{\prettyref{equ:P_cup_I}}=f(P)^\ast p(P).
\end{align}

We make the convention that for any $n\geq 3$,
\begin{align*} 
    P^n:=(((PP)P)\ldots P)P\quad\text{($n$ times)}.
\end{align*} The {\em Kleene star} and {\em plus} of $P$ is defined by
\begin{align}\label{equ:P^ast} 
    P^\ast:=\bigcup_{n\geq 0}P^n \quad\text{and}\quad P^+:=P^\ast P.
\end{align} For example, we have
\begin{align}\label{equ:I^ast} 
    I^\ast \stackrel{\prettyref{equ:IP=I}}=1\cup I \quad\text{and}\quad I^+=I,
\end{align} which shows
\begin{align}\label{equ:P_cup_I} 
    P\cup I \stackrel{\prettyref{equ:P_cup_I}}=(1\cup I)P \stackrel{\prettyref{equ:I^ast}}=I^\ast P.
\end{align} 

The {\em omega} of $P$ is defined by
\begin{align}\label{equ:P^omega} 
    P^\omega:=f(P^+)\stackrel{\prettyref{equ:f(P)}}=P^+\circ\emptyset.
\end{align} The least model of $P$ is captured by the omega of $P$ in the sense that \cite[Theorem 40]{Antic_i3}
\begin{align*} 
    LM(P)=P^\omega.
\end{align*}

\section{\texorpdfstring{$\mathcal{L}$}{L}-Relation}\label{sec:L}

We begin by studying \citeS{Green51} $\mathcal L$-relation, instantiated in the setting of propositional logic programs as follows. Given two programs $P$ and $R$, we have
\begin{align*} 
    P\leq_\mathcal{L} R \quad\text{iff}\quad P=QR,\quad\text{for some program $Q$.}
\end{align*} We then call $Q$ a {\em prefix} of $P$. In case $P\leq_\mathcal{L} R$ and $R\leq_\mathcal{L} P$, we say that $P$ and $R$ are {\em $\mathcal{L}$-equivalent} denoted by $P\equiv_\mathcal{L} R$. Roughly, we have $P\leq_\mathcal L R$ iff we can obtain the program $P$ from $R$ by manipulating the rules in $R$ via composition on the left. 

Since we cannot manipulate the rule bodies in $R$ via left composition, as an immediate consequence of \prettyref{equ:h(PR)}, we have
\begin{align*} 
    P\leq_\mathcal L R \quad\text{implies}\quad b(P)\subseteq b(R),
\end{align*} and, for every Krom program $K$,
\begin{align*} 
    P\leq_\mathcal L K \quad\text{implies}\quad P\text{ is Krom}.
\end{align*}

The following $\mathcal L$-relations are simple consequences from known facts listed in \prettyref{sec:Composition}, for every program $P$ and interpretations $I,J$:
\begin{align}\label{equ:L}
    P\leq_\mathcal L 1,\quad
    I\stackrel{\prettyref{equ:IP=I}}{\leq_\mathcal L} P,\quad
    I\stackrel{\prettyref{equ:IP=I}}{\equiv_\mathcal L} J,\quad
    ^IP\stackrel{\prettyref{equ:^IP}}{\leq_\mathcal L} P,\quad
    P\cup I\stackrel{\prettyref{equ:P_cup_I}}{\leq_\mathcal L} P,\quad
    P^+\stackrel{\prettyref{equ:P^ast}}{\leq_\mathcal L} P.
\end{align} The first relation shows that $1$ is $\mathcal L$-maximal, whereas the second relation shows that interpretations are $\mathcal L$-minimal, and the third relations shows that all interpretations are $\mathcal L$-equivalent. Moreover, we have
\begin{align*} 
    P\leq_\mathcal L I \quad\text{iff}\quad P=PI\stackrel{\prettyref{equ:T_P}}=T_P(I) \quad\text{iff}\quad \text{$P$ is an interpretation.}
\end{align*}

We have the following characterization of the $\mathcal L$-relation.

\begin{theorem} For any programs $P$ and $R$, we have $P\leq_\mathcal L R$ iff $P=QR$ for
\begin{align*} 
    Q:=\left\{h(r)\leftarrow h(S) \;\middle|\; 
    \begin{array}{l}
        r\in P\\
        S\subseteq_{sz(r)} R\\
        b(S)=b(r)\\
        \{h(r)\leftarrow h(S)\}R\subseteq P
    \end{array}\right\}.
\end{align*}
\end{theorem}
\begin{proof} The direction from right to left is trivial. For the other direction, we need to deduce $P=QR$ from $P\leq_\mathcal L R$. We first prove that $QR$ is a subset of $P$:
\begin{align*} 
    QR\stackrel{\prettyref{equ:bigcup}}=\bigcup_{s\in Q}(\{s\}R)=\bigcup_{\substack{r\in P,S\subseteq_{sz(r)} R\\b(S)=b(r)\\\{h(r)\leftarrow b(S)\}R\subseteq P}}\{h(r)\leftarrow h(S)\}R\;\subseteq P.
\end{align*} It remains to show that $P$ is a subset of $QR$. For this, we need to find for each rule $r\in P$ some rule $h(r)\leftarrow h(S)\in Q$ such that $S\subseteq_{sz(r)} R$, $b(S)=b(r)$, and $\{h(r)\leftarrow h(S)\}R\subseteq P$. By assumption, we have $P\leq_\mathcal L R$ and recall that $P\leq_\mathcal L R$ iff there is some program $Q'$ such that $P=Q'R$. This implies 
\begin{enumerate}
    \item $\{r\}=\{s\}S$ for some $s\in Q'$, $S\subseteq R$, with $h(s)=h(r)$, $b(S)=b(r)$, $b(s)=h(S)$ and, hence, $s=h(r)\leftarrow h(S)$; and
    \item $\{s\}R\subseteq P$.
\end{enumerate} This shows $s=h(r)\leftarrow h(S)\in Q$ and, hence, $r\in QR$.
\end{proof}

\section{\texorpdfstring{$\mathcal R$}{R}-Relation}\label{sec:R}

In this section, we study \citeS{Green51} $\mathcal R$-relation, instantiated in the setting of propositional logic programs as follows. Given two programs $P$ and $R$, we define
\begin{align*} 
    P\leq_\mathcal{R} R \quad\text{iff}\quad P=RS,\quad\text{for some program $S$.}
\end{align*} We then call $S$ a {\em suffix} of $P$. In case $P\leq_\mathcal{L} R$ and $R\leq_\mathcal{L} P$, we say that $P$ and $R$ are {\em $\mathcal{R}$-equivalent} denoted by $P\equiv_\mathcal{R} R$.

\begin{example} Consider the programs
\begin{align*} 
    \pi_{(a\,b)}:= \left\{
    \begin{array}{l}
        a\leftarrow b\\
        b\leftarrow a
    \end{array}
    \right\} \quad\text{and}\quad R:= \left\{
    \begin{array}{l}
        a\leftarrow b\\
        b\leftarrow b
    \end{array}
    \right\}.
\end{align*} 

We have
\begin{align*} 
    R=\pi_{(a\,b)}R
\end{align*} which shows
\begin{align*} 
    R\leq_\mathcal R \pi_{(a\,b)}.
\end{align*} 

On the other hand, there can be no program $S$ such that $\pi_{(a\,b)}=RS$ since we cannot rewrite the rule body $b$ of $R$ into $a$ and $b$ simultaneously via composition on the right. This shows
\begin{align*} 
    R<_\mathcal R\pi_{(a\,b)}.
\end{align*}
\end{example}

As an immediate consequence of \prettyref{equ:h(PR)}, we have
\begin{align}\label{equ:h(P)_subseteq_h(R)} 
    P\leq_\mathcal R R \quad\text{implies}\quad h(P)\subseteq h(R).
\end{align}

Given atoms $a,b\in A$, the identity
\begin{align*} 
    \{a\}=\{a\leftarrow b\}\{b\}
\end{align*} shows
\begin{align*} 
    \{a\}\leq_\mathcal R \{a\leftarrow b\}.
\end{align*} Since we cannot add body atoms to facts via composition on the right, we have
\begin{align*} 
    \{a\leftarrow b\}\not\leq_\mathcal R\{a\}.
\end{align*} Hence,
\begin{align*} 
    \{a\}<_\mathcal R\{a\leftarrow b\}.
\end{align*} On the other hand, given atoms $a_0,\ldots,a_k\in A$, $k\geq 1$, we have
\begin{align*} 
    \{a_0\leftarrow a_1\}&=\{a_0\leftarrow a_1,\ldots,a_k\}\{a_2,\ldots,a_k\}^\ominus\\
    \{a_0\leftarrow a_1,\ldots,a_k\}&=\{a_0\leftarrow a_1\}\{a_2,\ldots,a_k\}^\oplus
\end{align*} which shows
\begin{align*} 
    \{a_0\leftarrow a_1\}\equiv_\mathcal R\{a_0\leftarrow a_1,\ldots,a_k\}.
\end{align*} This means that all proper rules are with identitical head atom are $\mathcal R$-equivalent.

The following $\mathcal R$-relations are simple consequences from known facts listed in \prettyref{sec:Composition}, for every program $P$ and interpretation $I$:
\begin{align*} 
    &h(P)\stackrel{\prettyref{equ:h(P)}}{\leq_\mathcal R} P,\quad
    b(P)\stackrel{\prettyref{equ:h(P)}}{\leq_\mathcal R} p(P)^d,\quad
    P\stackrel{\prettyref{equ:f(P)^ast}}{\leq_\mathcal R} f(P)^\ast,\quad
    P\cup I \stackrel{\prettyref{equ:P_cup_I}}{\leq_\mathcal R} I^\ast,\quad
    f(P)\stackrel{\prettyref{equ:f(P)}}{\leq_\mathcal R} P,\\
    &PI^\oplus\leq_\mathcal R P,\quad
    PI^\ominus\leq_\mathcal R P,\quad
    P^I\stackrel{\prettyref{equ:^IP}}\leq_\mathcal R P,\quad
    T_P(I)\stackrel{\prettyref{equ:T_P}}{\leq_\mathcal R} P,\quad
    I\cap J\stackrel{\prettyref{equ:I_cap_J}}{\leq_\mathcal R} 1^I, 1^J,\quad
    P^\omega\stackrel{\prettyref{equ:P^omega}}{\leq_\mathcal R} P.
\end{align*} Moreover, we have
\begin{align*} 
    P&\leq_\mathcal R I \quad\text{iff}\quad P=IS\stackrel{\prettyref{equ:IP=I}}=I,
\end{align*} which means that
\begin{align*} 
    I\equiv_\mathcal R J \quad\text{iff}\quad I=J.
\end{align*}

We have the following characterization of the $\mathcal R$-relation.

\begin{theorem} For any programs $P$ and $R$, we have $P\leq_\mathcal R R$ iff for each rule $r\in P$ there is a rule $s\in R$ and a program $S_r$ such that $\{r\}=\{s\}S_r$ and $RS_r\subseteq P$, in which case we have $P=RS$ with $S:=\bigcup_{r\in P}S_r$.
\end{theorem}
\begin{proof} The direction from right to left holds trivially. For the other direction, assume $P\leq_\mathcal R R$, which is equivalent to $P=RS'$, for some program $S'$. This means that for each rule $r\in P$ there is a rule $s\in R$ and a subset $S_r$ of size $sz(s)$ of $S'$ such that $\{r\}=\{s\}S_r$ and $RS_r\subseteq P$.
\end{proof}

\section{\texorpdfstring{$\mathcal J$}{J}-Relation}\label{sec:J}

In this section, we study \citeS{Green51} $\mathcal J$-relation defined in the setting of propositional logic programs as follows. Given two programs $P$ and $R$, we define\footnote{Recall that composition is not associative, which means that we cannot omit the brackets.}
\begin{align*} 
    P\leq_\mathcal J R\quad\text{iff}\quad P=(QR)S,\quad\text{for some programs $Q,S$.}
\end{align*} We then call $Q$ a {\em prefix} and $S$ a {\em suffix} of $P$. In case $P\leq_\mathcal J R$ and $R\leq_\mathcal J P$, we say that $P$ and $R$ are {\em $\mathcal J$-equivalent} denoted by $P\equiv_\mathcal J R$.



The $\mathcal J$-relation coincides with the propositional version of \citeS[Definition 5]{Antic_i8} similarity relation and most of the results in this section are instances of \citeS{Antic_i8} more general results which we repeat here for completeness.

Of course,
\begin{align}\label{equ:L_R-J} 
    P\leq_\mathcal L R \quad\text{or}\quad P\leq_\mathcal R R \quad\text{implies}\quad P\leq_\mathcal J R,
\end{align} which means that the $\mathcal L$-relations of \prettyref{sec:L} and $\mathcal R$-relations of \prettyref{sec:R} yield $\mathcal J$-relations.

We have
\begin{align*} 
    P\equiv_\mathcal J I \quad&\text{iff}\quad P=(QI)S \stackrel{\prettyref{equ:T_P}}=(T_Q(I))S \stackrel{\prettyref{equ:IP=I}}=T_Q(I),\quad\text{for some programs $Q,S$},\\ 
        &\text{iff}\quad \text{$P$ is an interpretation.}
\end{align*} In particular, this means that all interpretations are $\mathcal J$-equivalent (which already follows from the fact that all interpretations are $\mathcal L$-equivalent; see \prettyref{equ:L}). Moreover, we have
\begin{align}\label{equ:f} 
    P\equiv_\mathcal J f(P) \quad&\text{iff}\quad P=f(P),
\end{align} which is shown as follows. We have
\begin{align}\label{equ:Q_circ_f(P)} 
    P\leq_\mathcal J f(P) 
        \quad&\text{iff}\quad P=(Q f(P))S \stackrel{\prettyref{equ:T_P},\prettyref{equ:IP=I}}=T_Q(f(P)),\quad\text{for some programs $Q,S$,}\\
        &\text{iff}\quad P\subseteq f(P),
\end{align} which together with $f(P)\subseteq P$ yields $P=f(P)$.

We have the following characterization of the $\mathcal J$-relation.

\begin{fact}\label{thm:J} For any programs $P$ and $R$, we have $P\leq_\mathcal J R$ iff for each rule $r\in P$ there is a rule $s_r$, a subset $S_{r,R}$ of $R$, and a program $S_r$ such that
\begin{align*} 
    \{r\}=(\{s_r\}S_{r,R})S_r \quad\text{and}\quad (\{s_r\}R)S_r\subseteq P.
\end{align*} In this case, we have $P=(QR)S$ with $Q:=\bigcup_{r\in P}\{s_r\}$ and $S:=\bigcup_{r\in P}S_r$. 
\end{fact}



\begin{example}[\citeA{Antic_i8}, Example 13] Consider the programs
\begin{align*} 
    P=\left\{
    \begin{array}{l}
        c\\
        a\leftarrow b,c\\
        b\leftarrow a,c
    \end{array}
    \right\} \quad\text{and}\quad \pi_{(a\,b)}=\left\{
    \begin{array}{l}
        a\leftarrow b\\
        b\leftarrow a
    \end{array}
    \right\}.
\end{align*} 

We construct the programs $Q$ and $S$ such that $P=(Q\pi_{(a\,b)})S$ according to \prettyref{thm:J}. Define
\begin{align*} 
    &r_1:=c \quad\text{implies}\quad s_{r_1}:=c \quad\text{and}\quad S_{r_1}:=\emptyset,\\
    &r_2:=a\leftarrow b,c \quad\text{implies}\quad s_{r_2}:=a\leftarrow a \quad\text{and}\quad S_{r_2}:=\{b\leftarrow b,c\},\\
    &r_3:=b\leftarrow a,c \quad\text{implies}\quad s_{r_3}:=b\leftarrow b \quad\text{and}\quad S_{r_3}:=\{a\leftarrow a,c\},
\end{align*} and
\begin{align*} 
    Q:=\{s_{r_1},s_{r_2},s_{r_3}\}=\{c\}^\ast \quad\text{and}\quad S:=S_{r_1}\cup S_{r_2}\cup S_{r_3}=\{c\}^\oplus-1^{\{c\}}.
\end{align*} This yields
\begin{align*} 
    P=\{c\}^\ast\pi_{(a\,b)}(\{c\}^\oplus-1^{\{c\}}) \quad\text{and}\quad \pi_{(a\,b)}=1^{\{a,b\}}P\{c\}^\ominus
\end{align*} which shows
\begin{align*} 
    P\equiv_\mathcal J\pi_{(a\,b)}.
\end{align*}
\end{example}

The following example shows that $\mathcal J$-equivalence and logical equivalence are ``orthogonal'' concepts \cite<cf.>[Example 15]{Antic_i8}.

\begin{example} The empty program is logically equivalent with respect to the least model semantics to the propositional program $a\leftarrow a$ consisting of a single rule. Since we cannot obtain the rule $a\leftarrow a$ from the empty program via composition, logical equivalence does not imply $\mathcal J$-equivalence. 

For the other direction, the computations
\begin{align*} 
    \left\{
    \begin{array}{l}
        a\\
        b\leftarrow a
    \end{array}
    \right\}= \left\{
    \begin{array}{l}
        a\\
        b\leftarrow a,b
    \end{array}
    \right\}\{b\}^\ominus \quad\text{and}\quad \left\{
    \begin{array}{l}
        a\\
        b\leftarrow a,b
    \end{array}
    \right\}=\left\{
    \begin{array}{l}
        a\\
        b\leftarrow a
    \end{array}
    \right\}\{b\}^\oplus
\end{align*} show that the programs $\left\{
    \begin{array}{l}
        a\\
        b\leftarrow a       
    \end{array}
    \right\}$ and $\left\{
    \begin{array}{l}
        a\\
        b\leftarrow a,b     
    \end{array}
    \right\}$ are $\mathcal J$-equivalent; however, they are obviously not logically equivalent.
\end{example}





\section{Future work}

The major line of future research is to lift the concepts and results of this paper from propostional to first-order \cite{Lloyd87,Apt90} and higher-order logic programs \cite{Chen93,Miller12}, and to answer set programs \cite{Gelfond91} \cite<cf.>{Lifschitz19,Brewka11a} possibly containing negation as failure \cite{Clark78} in rule bodies or disjunctions in rule heads \cite{Lobo92,Eiter97}. The sequential composition of first-order logic programs has been defined in \citeA[§3.1]{Antic_i2}, and the composition of answer set programs in \citeA{Antic_i6}. The definition of sequential composition of higher-order or disjunctive logic programs is missing.

\section{Conclusion}

This paper analyzed the sequential decomposition of propositional logic programs by studying Green's relations $\mathcal{L,R,J}$---well-known in semigroup theory---in the context of logic programming. In a broader sense, this paper is a further step towards an algebraic theory of logic programming.

\if\isdraft1\newpage\fi
\bibliographystyle{theapa}
\bibliography{/Users/christianantic/Bibdesk/Bibliography,/Users/christianantic/Bibdesk/Preprints,/Users/christianantic/Bibdesk/Publications}
\end{document}